\documentclass[11pt,a4paper]{article}
\usepackage[T1]{fontenc}
\usepackage[utf8]{inputenc}
\usepackage{lmodern}
\usepackage{microtype}
\usepackage[margin=1in]{geometry}
\usepackage{cite}
\usepackage{amsmath,amssymb,amsfonts,amsthm,bm}
\usepackage{graphicx,booktabs,textcomp,xcolor}
\usepackage[hidelinks]{hyperref}
\graphicspath{{figures/}}
\newtheorem{theorem}{Theorem}
\newtheorem{proposition}[theorem]{Proposition}
\newtheorem{lemma}[theorem]{Lemma}
\newtheorem{corollary}[theorem]{Corollary}
\theoremstyle{definition}
\newtheorem{definition}[theorem]{Definition}
\newtheorem{assumption}[theorem]{Assumption}
\theoremstyle{remark}
\newtheorem{remark}[theorem]{Remark}
\title{Hierarchical Learning Automata Without a Shared Reinforcement
Signal: Recovering the Environment Response From Action-Probability Changes}
\author{Joss Armstrong\\
NM Research, Ericsson, Athlone, Ireland\\
\small\href{mailto:joss.armstrong@ericsson.com}{joss.armstrong@ericsson.com}}
\date{13 October 2026}
\begin{document}
\maketitle
\begin{abstract}
A hierarchy of learning automata in the sense of Thathachar and
Ramakrishnan activates one root-to-leaf path per round, and every
automaton on that path updates from the environment response. Analyses of
this hierarchy deliver the response to every level.
In many deployed hierarchies only the root observes it. The interfaces
below carry requests but no feedback channel. This paper shows that
the response can be recovered at every level without one. An active child
receives one bit, the sign of the change in the action probability its
parent assigns to it, and reads the response from that sign. For any update rule in which the selected action's
probability rises on reward and does not rise on penalty, the decoded bit
equals the response at every depth. This class includes linear
reward-penalty, reward-epsilon-penalty and reward-inaction. The protocol
hierarchy and the broadcast hierarchy therefore generate identical sample
paths, so any convergence result proved for the broadcast hierarchy under
these rules holds for the protocol hierarchy as well. The paper then analyses a proportional reward-penalty
rule that keeps the allocation interior. For two actions it coincides with
symmetric linear reward-penalty. For three or more it lies outside the
classical schemes and balances each action's failure rate against its
remaining allocation, not its penalty rate. Its mean
flow converges globally to a closed-form allocation and drops actions
below a computable threshold. High-precision
replays on hierarchies of up to 16,384 leaves show zero decoding errors.
\end{abstract}

\noindent\textbf{Keywords:} Distributed decision making, hierarchical learning automata,
hierarchical systems, learning automata, reinforcement learning,
stochastic approximation, variable-structure learning automata.

\section{Introduction}
\label{sec:intro}
% ==================================================================

Consider a platform that routes each user request to one of
several third-party model providers. Each provider routes the request
across models it serves on behalf of other vendors, and each of those
picks a serving configuration. Every interface in this chain carries the
request and nothing that evaluates it. The platform scores each request
by a rule of its own and adjusts its allocation to the providers. The
only thing any level can learn about its standing is whether its parent's
allocation to it rose or fell. No party exposes its internals to the
party above. The same structure arises when an orchestrator selects among
vendor agents that select among licensed models, when a federation
delegates work to member services, and inside one organisation whenever
adjacent levels are owned by different teams and no feedback interface
was built.

The learning-automata literature studied hierarchical selection under a
single binary response decades ago. Thathachar and
Ramakrishnan~\cite{ThathacharRamakrishnan1981} arranged automata in a tree,
with one path activated per round. Thathachar and
Phansalkar~\cite{ThathacharPhansalkar1995} developed an ODE analysis of
hierarchies in which all levels adapt at once under linear
reward-inaction. Teams of automata~\cite{ThathacharRamakrishnan1982} and
the later networks surveyed by Thathachar and
Sastry~\cite{ThathacharSastry2002} share the same informational premise.
The environment response is made available to every automaton that
updates. In the survey's words, ``the updating of all FALA [finite action-set learning automata] in all
units is based on the same reinforcement signal''~\cite{ThathacharSastry2002}.

Deployed hierarchies often break that premise. The response exists where
the outcome is measured, at the root, and no interface below carries it.
The reasons vary. The levels may belong to different parties. The
interface between them may be fixed by a published API or a standard,
with no field for a score. The response may arrive late and out of band,
after the request path has closed. Or the channel was never built. What
every level can learn is whether its parent's allocation to it rose or
fell. The question is whether that one bit is enough feedback for
learning.

This paper shows that the sign of the change in the share is enough.
For an interior probability vector, the linear reward-penalty and
reward-inaction rules increase the selected child's probability on reward
and do not increase it on penalty. An active child can therefore read the
response from the sign of the change in its own share. No separate
response message is needed. The decoded bit is exact at every depth, so
the protocol hierarchy is the broadcast hierarchy, sample path by sample
path. The root may score the outcome in any form and update by any rule.
What propagates is the bit decoded at depth one (Remark~\ref{rem:root-free}).

\subsection{Contributions}

\begin{enumerate}
\item \textbf{A broadcast-free protocol} (Definition~\ref{def:protocol})
and a class of \emph{sign-separating} update rules
(Definition~\ref{def:sign-sep}) for which it is exact. For these rules the
child's decoded bit equals the root response at every depth and for any
locally chosen context maps (Theorem~\ref{thm:signal},
Remark~\ref{rem:maps}), and the sign is a sufficient statistic for the
response (Corollary~\ref{cor:sufficiency}). The class contains linear
reward-penalty, reward-$\epsilon$-penalty and reward-inaction.

\item \textbf{Pathwise equivalence} (Proposition~\ref{prop:equivalence}).
Under common random numbers the protocol hierarchy and the broadcast
hierarchy have identical trajectories. Convergence results proved for
broadcast hierarchies under sign-separating rules therefore hold without
the broadcast, under the same hypotheses. In particular, the ODE approximation and local
convergence results in~\cite{ThathacharPhansalkar1995} carry over to its
binary-response selector-tree specialization
(Corollary~\ref{cor:tp95}).

\item \textbf{An interior-allocation rule.} The paper analyses a
proportional reward-penalty (PRP) rule
(\eqref{eq:reward}, \eqref{eq:neg-selected} and~\eqref{eq:neg-sibling}). It preserves the
simplex with no projection (Theorem~\ref{thm:integrity}). For two actions
it coincides with symmetric linear reward-penalty
(Proposition~\ref{prop:classical}). For three or more actions it differs.
Linear reward-penalty equalises the penalty rate $w_i(1-p_i)$, while PRP
equalises the failure rate per unit of slack, $(1-p_i)/(1-w_i)$. The
paper gives the unique interior equilibrium in closed form
(Theorem~\ref{thm:equilibrium}), almost-sure convergence under decreasing steps, for two actions
(Theorem~\ref{thm:stochN2}) and for general $N$ on the event that the
weights stay away from the boundary (Corollary~\ref{cor:stochN}),
explicit spectral bounds
(Theorem~\ref{thm:stability}) and global convergence of the mean flow to
the minimiser of a Kullback--Leibler potential, with exact dropping of
actions below a threshold (Theorem~\ref{thm:G7}).

\item \textbf{Composition and fault response.} Along the active path each
node runs its standalone transition law on a subsequence selected by its
ancestors (Theorem~\ref{thm:hierarchy}). An intermediary that is stuck at
one value or that inverts the bit delivers strictly lower mean quality
than one that passes it on, by an amount set by the dispersion of its
children's qualities (Corollaries~\ref{lem:qual-ordering}
and~\ref{cor:dispersion}).

\item \textbf{Numerical results.} High-precision replays on synthetic
hierarchies up to 16,384 leaves, and float64 runs on three public
datasets, record zero decoding errors. A nonstationary tracking study
compares PRP, linear reward-penalty and reward-inaction under the
protocol, and classifies every finite-precision decoding error
(Section~\ref{sec:experiments}).
\end{enumerate}

\subsection{What is not claimed}

Three limits apply. First, convergence of the PRP rule is proved for a
single selector with fixed child response probabilities. Joint
convergence of a PRP hierarchy whose levels adapt at once is open
(Section~\ref{sec:discussion}). Second, PRP is not claimed to be a better
tracker than reward-inaction. In the experiments of
Section~\ref{sec:experiments} it is not. Third, the protocol needs the exact
sign of the parent's update (condition C4). A child that infers its share
from a noisy traffic count is not covered.

% ==================================================================
\section{Related Work}
\label{sec:related}
% ==================================================================

\subsection{Learning automata}
A variable-structure learning automaton keeps a probability vector over
its actions and updates it from a binary (P-model) or bounded (S-model)
environment response~\cite{NarendraThathachar1974,NarendraThathachar1989}.
Recent monographs survey the field~\cite{Rezvanian2018,ZhangZhou2023}.
Linear reward-penalty schemes were analysed for automata by
McLaren~\cite{McLaren1966} and by Chandrasekaran and
Shen~\cite{ChandrasekaranShen1968}. On penalty they spread the lost
mass uniformly over the other actions, and their expected limit is
interior. Lakshmivarahan and Thathachar~\cite{LakshmivarahanThathachar1973}
characterised the absolutely expedient nonlinear reward-penalty schemes.
In that class the other actions change in proportion to their
probabilities, with a gain that is common to all actions and does not
depend on which action was chosen. Absolutely expedient schemes are
$\epsilon$-optimal and converge towards a vertex. Estimator and pursuit
algorithms~\cite{ThathacharSastry2002,Zhang2015} use running estimates of
each action's reward probability, and parameter-free variants remove the
step-size tuning~\cite{Di2021}.

PRP sits outside these classes for three or more actions. Its penalty
step scales the other actions by a factor that depends on the share of the
chosen action. Its rest point is interior, so it is not absolutely
expedient. For two actions it reduces to symmetric linear reward-penalty,
and its rest point is McLaren's~\cite{McLaren1966}
(Proposition~\ref{prop:classical}).

\subsection{Hierarchies and teams of automata}
Thathachar and Ramakrishnan~\cite{ThathacharRamakrishnan1981} introduced
the hierarchical system with one active path per round and a coupled step
size between levels. They proved joint $\epsilon$-optimality. The
probability of a leaf is the product of the action probabilities on its
path, which is part~(c) of Theorem~\ref{thm:hierarchy} here. Teams of
automata play a common-payoff game~\cite{ThathacharRamakrishnan1982}.
Decentralised automata converge to Nash equilibria of stochastic games
under incomplete information~\cite{Sastry1994}, and reward-inaction with
artificial barriers reaches mixed equilibria~\cite{Yazidi2023,Hassan2023}.
Thathachar and Phansalkar~\cite{ThathacharPhansalkar1995} analysed
fixed-step reward-inaction hierarchies through a small-step ODE limit,
obtaining local convergence results under an isolated-maxima assumption.
In all of these works every automaton
that updates receives the environment response directly. More recent
hierarchies of two-action pursuit automata, with every action at a leaf,
are proved $\epsilon$-optimal for large action
sets~\cite{Yazidi2020,Omslandseter2024,Omslandseter2023}, and a hierarchy of twofold
resource allocation automata solves stochastic nonlinear allocation
problems~\cite{GranmoOommen2010}. Pursuit updates move towards the action
with the best running estimate rather than the selected action, so they
are not sign-separating and Proposition~\ref{prop:equivalence} does not
cover them (Section~\ref{sec:discussion}). We have not found a scheme in
which a child recovers the response from its own probability change. For selector trees covered by
Proposition~\ref{prop:equivalence}, the protocol of
Section~\ref{sec:protocol} reproduces the probability process of these
schemes with no response delivered below the root.

\subsection{Bandits and multiplicative weights}
A Bernoulli bandit~\cite{Bubeck2012} also learns from one bit per round,
so the bit count alone separates nothing. Hedge and multiplicative
weights~\cite{Littlestone1994,Freund1997} need a full loss vector each
round. EXP3~\cite{Auer2002b} needs an importance-weighted loss for the
selected arm. All of these assume that the learner receives a reward
signal. Here each intermediate learner reconstructs the response from
the sign of its own allocation change, without a reward message.
A contextual bandit~\cite[Sec.~4]{Bubeck2012} observes a feature vector in
a space it shares with its comparator policies, and the associative
reward-penalty automaton~\cite{BartoAnandan1985} receives a context with
each trial. Here each selector's context is its own and is never
transmitted (Remark~\ref{rem:maps}).
$\mathcal{X}$-armed bandits~\cite{Bubeck2011} partition an action space
hierarchically, but a single learner navigates the tree.
Table~\ref{tab:learning-rules} sets out the information each approach
requires.

\begin{table}[t]
\caption{Feedback required per round. The last column counts explicit
feedback in addition to the learner's existing observations. This work
uses one bit derived from the parent's own update (C4).}
\label{tab:learning-rules}
\centering
\footnotesize
\begin{tabular}{@{}lll@{}}
\toprule
Approach & Observed by learner & Feedback \\
\midrule
Hedge / MW~\cite{Littlestone1994,Freund1997} & loss vector & $N$ reals \\
EXP3~\cite{Auer2002b} & IW loss, selected arm & 1 real \\
Bernoulli bandit~\cite{Bubeck2012} & binary reward & 1 bit \\
LA hierarchy~\cite{ThathacharRamakrishnan1981,ThathacharPhansalkar1995} & response & 1 bit or 1 real \\
\midrule
\textbf{This work} & sign of own share change & \textbf{1 bit, implicit} \\
\bottomrule
\end{tabular}
\end{table}

\subsection{Hierarchical credit assignment}
Feudal reinforcement learning~\cite{Dayan1992} and FeUdal
Networks~\cite{Vezhnevets2017} pass explicit signals between levels.
Samejima et al.~\cite{Samejima2003} extract credit from changes in gating
weights. That is close in spirit to the present protocol, but it relies
on value-function estimates and centralised training. Mixture-of-experts
gating~\cite{Jacobs1991,Shazeer2017} needs a differentiable loss.
COMA~\cite{Foerster2018} and LICA~\cite{Zhou2020} address credit among
simultaneous agents, which the present model excludes (condition C2).

\subsection{Replicator dynamics}
The PRP mean drift has replicator form~\cite{TaylorJonker1978,Weibull1995}
with a cost $h_i=(1-p_i)/(1-w_i)$ that depends only on the child's own
weight. Such systems admit a potential, and Theorem~\ref{thm:G7} uses it
in the Lyapunov tradition of~\cite{HofbauerSigmund1998}. Börgers and
Sarin~\cite{BorgersSarin1997} showed that the expected motion of Cross's
learning rule~\cite{Cross1973} is the replicator equation. In that
tradition the learner observes its own payoff. Here it does not.

% ==================================================================
\section{System Model}
\label{sec:model}
% ==================================================================

\subsection{Hierarchy and information constraints}

A hierarchy is a rooted tree $\mathcal{T}=(V,E)$. Each node is a selector
or a leaf. A selector $v$ has $N_v\ge 2$ children
$\mathrm{ch}(v)=\{c_1,\ldots,c_{N_v}\}$. The root $r$ has no parent. In the
running example the root is the platform, the depth-one selectors are
providers, the depth-two selectors are models, and the leaves are serving
configurations.

At each round $t$ an input $x_t$ is drawn independently from a fixed
distribution $\mathcal{D}$ on a space $\mathcal{X}$. Each leaf $\ell$ has
a quality function $q_\ell:\mathcal{X}\to[0,1]$ that no other node knows.
Each selector $v$ has a context map $\kappa_v:\mathcal{X}\to\mathcal{K}_v$
onto a finite set. The input $x_t$ carries the request and whatever else
any selector observes at round $t$, and $\kappa_v$ may read any part of
it that no other selector reads, or none of it. Each selector chooses
its own map. Theorem~\ref{thm:signal}, Corollary~\ref{cor:sufficiency}
and Proposition~\ref{prop:equivalence} hold for any such maps
(Remark~\ref{rem:maps}). The single-selector results of
Section~\ref{sec:prp} apply to one cell of one selector whenever its
child response probabilities are fixed on that cell's active
subsequence. When every selector uses the same map, the rounds of one
context form a hierarchy of the kind in Corollary~\ref{cor:tp95}, and it
applies per context. The context superscript is suppressed where it
plays no role.

Four conditions define the setting.
\begin{description}
\item[C1] \emph{Opacity.} A component's quality and internal state are not
observable by any selector.
\item[C2] \emph{Single path.} Each active selector selects exactly one child.
Only the selected path produces an outcome. Unselected alternatives remain
counterfactual. An outcome produced jointly by several siblings is not
covered.
\item[C3] \emph{Local updating.} A selector updates using only its own weight
vector, the identity of its selected child, and its signal.
\item[C4] \emph{Sign observability.} On each round in which it is active,
a non-root selector learns one bit, the sign of the change in the weight
its parent assigns to it. The bit is derived from the parent's update,
not estimated from realised traffic.
\end{description}
Under C1--C4 the only evaluation information that crosses a parent--child
boundary is that one bit.

\subsection{Selection and response}

Each selector $v$ is a variable-structure learning automaton. It keeps,
for each context $k$, a weight vector
$\mathbf{w}_v^k\in\Delta^{N_v-1}$, its action probability vector,
initialised uniformly.

\begin{definition}[Active path]\label{def:active-path}
At round $t$ the root selects a child with probability given by its
weights, that child selects one of its own children in the same way, and
so on to a leaf. A node is \emph{active} at $t$ if it lies on this path.
Write $A_{v,t}$ for this event. A node observes $A_{v,t}$ locally, since
it is invoked exactly when its parent selects it.
\end{definition}

Selector $v$ on the active path computes $k=\kappa_v(x_t)$ and selects
$c_t$ with
\begin{equation}\label{eq:selection}
P(c_t=c_i)=w_{v,i}^k(t).
\end{equation}
An implementation may mix a small uniform exploration probability into
this selection. The protocol uses only the identity of the selected
child, so nothing below depends on how it was selected.
Only the root observes the response $o_t\in\{0,1\}$, a Bernoulli
variable whose success probability depends on the selected leaf and the
input. One construction thresholds a noisy leaf quality,
$o_t=\mathbf{1}[\tilde q(x_t)>\theta]$. Section~\ref{sec:experiments}
states the construction each experiment uses.

\subsection{Update rules}

An update rule $U$ maps a weight vector $w$, a selected index $i$ and a
bit $b$ to a new weight vector $w'=U(w,i,b)$. The following rules are used
throughout, with step $\eta\in(0,1)$. All share the reward step
\begin{equation}\label{eq:reward}
w'_i=(1-\eta)w_i+\eta,\qquad w'_j=(1-\eta)w_j\ (j\ne i).
\end{equation}
They differ on penalty.
\begin{itemize}
\item \emph{Linear reward-inaction} ($L_{R-I}$) leaves $w$ unchanged.
\item \emph{Linear reward-penalty} ($L_{R-P}$) sets $w'_i=(1-\eta)w_i$ and
$w'_j=(1-\eta)w_j+\eta/(N-1)$. The reward-$\epsilon$-penalty variant
uses step $\epsilon\eta$ on penalty.
\item \emph{Proportional reward-penalty} (PRP) sets
\begin{align}
w'_i &= (1-\eta)\,w_i, \label{eq:neg-selected}\\
w'_j &= w_j\,\frac{1-w_i+\eta\,w_i}{1-w_i},\quad j\ne i.
\label{eq:neg-sibling}
\end{align}
The mass removed from the selected child is returned to the others in
proportion to their weights.
\end{itemize}

% ==================================================================
\section{The Broadcast-Free Protocol}
\label{sec:protocol}
% ==================================================================

\begin{definition}[Sign-separating rule]\label{def:sign-sep}
An update rule $U$ is \emph{sign-separating} if, for every interior $w$
and every $i$, the selected coordinate satisfies $w'_i>w_i$ when $b=1$ and
$w'_i\le w_i$ when $b=0$.
\end{definition}

$L_{R-I}$, $L_{R-P}$, reward-$\epsilon$-penalty and PRP are all
sign-separating. Under the reward step~\eqref{eq:reward}
$w'_i-w_i=\eta(1-w_i)>0$ for interior $w$. On penalty the change is $0$
for $L_{R-I}$ and $-\eta w_i$ or $-\epsilon\eta w_i$ for the others.
Estimator and pursuit algorithms are not sign-separating in general. Their
update moves towards the action with the best running estimate, which need
not be the selected action.

\begin{definition}[Broadcast-free protocol]\label{def:protocol}
In each round:
\begin{enumerate}
\item The root updates its weights with $b=o_t$.
\item A non-root selector $v$ checks $A_{v,t}$. If $v$ is not active it
does nothing, and $\mathbf{w}_v(t+1)=\mathbf{w}_v(t)$.
\item If $v$ is active, it receives $s_v(t)=\mathbf{1}[\delta_v(t)>0]$,
where $\delta_v(t)=w_{p(v),v}(t+1)-w_{p(v),v}(t)$ is the change in its own
share and $p(v)$ is its parent.
\item It updates its own weights with $b=s_v(t)$.
\end{enumerate}
\end{definition}
Each selector may use a different sign-separating rule. The response
$o_t$ is never sent below the root.

\begin{assumption}[Protocol in place]\label{asm:adoption}
Every selector runs Definition~\ref{def:protocol}. The operator of the
hierarchy imposes this, in the way a platform fixes the interface its
providers must implement. The paper does not model the decision to adopt.
\end{assumption}

\begin{theorem}[Exact decoding]\label{thm:signal}
Suppose every selector uses a sign-separating rule and every parent's
weight vector is interior. If $v$ is active at round $t$, then
$s_v(t)=o_t$, at every depth.
\end{theorem}

\begin{proof}
By induction on depth. The root uses $b=o_t$. Suppose the active parent
$p(v)$ updates with $b=o_t$, and $v$ is its selected child. If $o_t=1$,
sign separation gives $\delta_v>0$ and $s_v=1$. If $o_t=0$, it gives
$\delta_v\le 0$ and $s_v=0$. So $v$ updates with $b=o_t$, and the
induction continues.
\end{proof}

\begin{remark}[The root's rule is free]\label{rem:root-free}
Theorem~\ref{thm:signal} uses the root only to supply the first bit. If
the root scores the outcome in some other form and updates by any rule,
sign-separating or not, an active depth-one child still sets
$s_v(t)=\mathbf{1}[\delta_v(t)>0]$, and the induction from depth one
downward is unchanged. Every deeper active node then decodes that
depth-one bit exactly. The binary response at the root is the
learning-automata convention, not a requirement of the protocol.
\end{remark}

\begin{remark}[The activity gate is needed]\label{rem:activity-gate}
An active parent changes every child's weight, not only the selected
one. Under PRP an unselected sibling moves opposite to the selected
child, so a sibling that decoded would read $1-o_t$. Under $L_{R-P}$ the
sign of an unselected sibling's change depends on its weight and carries
no fixed relation to $o_t$. The gate in step~2 removes both problems, and
it costs nothing, since $A_{v,t}$ is observed locally.
\end{remark}

\begin{corollary}[Sufficiency of the sign]\label{cor:sufficiency}
Under the conditions of Theorem~\ref{thm:signal},
$I(o_t;s_v(t)\mid A_{v,t})=H(o_t\mid A_{v,t})$. If the rule is a
deterministic map of $(w,i,b)$, as all rules above are, the magnitude
adds nothing:
$I(o_t;|\delta_v(t)|\mid s_v(t),\mathcal{F}_{t-1},A_{v,t})=0$.
\end{corollary}

\begin{proof}
Since $s_v=o_t$, $H(o_t\mid s_v,A_{v,t})=0$. Given the pre-round state
and the selected index, $\delta_v$ is a function of $b=s_v$, so
$|\delta_v|$ is measurable with respect to
$\mathcal{F}_{t-1}\vee\sigma(s_v)$.
\end{proof}

\begin{proposition}[Pathwise equivalence]\label{prop:equivalence}
Fix a hierarchy in which every selector uses a sign-separating rule and
starts from an interior vector that the rule keeps interior. Run it twice
on one probability space, with the same inputs, the same selection
variates at every node and the same response variates. In the
\emph{broadcast} run every active selector updates with $b=o_t$. In the
\emph{protocol} run selectors follow Definition~\ref{def:protocol}. Then
in exact arithmetic the two runs have identical weight vectors at every
node and every round.
\end{proposition}

\begin{proof}
By induction on $t$. If the weights agree at $t$, the two runs select the
same path and observe the same $o_t$. By Theorem~\ref{thm:signal} every
active selector in the protocol run updates with $b=o_t$, as in the
broadcast run. Inactive selectors do not update in either run. So the
weights agree at $t+1$.
\end{proof}

\begin{remark}[Context maps are not read]\label{rem:maps}
No step of Definition~\ref{def:protocol} evaluates a context map. The
parent updates the vector of its operative context, and the child learns
the sign of its own coordinate in that vector. The child's map only
selects which of its own vectors it updates. Theorem~\ref{thm:signal}
and Proposition~\ref{prop:equivalence} therefore hold when the maps
differ across levels, and when a selector changes its map at any round
by any rule of its own observations, since the broadcast run changes it
identically. A protocol selector holds exactly the information of a
broadcast selector. It conditions on context as the associative
reward-penalty automaton does~\cite{BartoAnandan1985}, with one vector
per cell of a partition that is its own and that no other selector sees.
A finer partition cannot lower the best attainable response, the
comparator in contextual bandit analysis~\cite[Sec.~4]{Bubeck2012},
since the finer policy class contains the coarser one. How finely to
partition a finite sample is the selector's own trade-off and is outside
this paper.
\end{remark}

Proposition~\ref{prop:equivalence} turns any theorem about a broadcast
hierarchy with sign-separating rules into a theorem about the
broadcast-free hierarchy. The following case is the most useful.

\begin{corollary}[Reward-inaction hierarchies]\label{cor:tp95}
Consider the single-unit tree specialization of Thathachar and
Phansalkar~\cite{ThathacharPhansalkar1995}, with one automaton per subunit
and each action selecting one child. Let every automaton start in the
interior and use $L_{R-I}$ with a fixed step in $(0,1)$, and restrict the
environment response to $\{0,1\}$. Under C4, replacing the response at
each non-root selector by its decoded bit $s_v(t)$ preserves the process
pathwise. Consequently, the small-step ODE approximation of their
Theorem~1 transfers unchanged under its hypotheses. Their local stability
and convergence conclusions also transfer whenever their hypotheses hold,
including Assumption~1 that the expected-reward maxima are isolated.
\end{corollary}

\begin{proof}
$L_{R-I}$ keeps interior vectors interior for $\eta<1$, since
$(1-\eta)w_i+\eta<1$ whenever $w_i<1$. It is sign-separating. In the
cited hierarchy, Algorithm~(9) applies $L_{R-I}$ to automata in active
subunits and leaves the others unchanged. With one automaton per subunit,
this is the activity gate in Definition~\ref{def:protocol}.
Proposition~\ref{prop:equivalence} applies.
\end{proof}

\begin{remark}[Bounded responses]\label{rem:s-model}
The cited analysis also allows a response $r\in[0,1]$. One bit does not
carry it, and the corollary is stated for the binary case only.
\end{remark}

\begin{remark}[Finite precision]\label{rem:precision}
Theorem~\ref{thm:signal} is a statement about sign-separating updates on
the interior in exact arithmetic, where a zero change arises only under
$L_{R-I}$ on penalty and decoding it as a penalty is correct. In floating
point the reward increment $\eta(1-w_i)$ is lost to rounding once the
share is close enough to one, and a clamp to a weight floor followed by
renormalisation can cancel or reverse an update. The update as applied is
then not sign-separating, although the represented state is still
interior. An implementation therefore keeps the applied update
sign-separating. Reward-inaction
with artificial barriers~\cite{Hassan2023} replaces the vertices $1$ and
$0$ in the reward step by $p_{\max}<1$ and $p_{\min}>0$ and needs no
clamp. Section~\ref{sec:experiments} states what each experiment does.
\end{remark}

\subsection{A two-level example}
\label{sec:example}

Let the root $r$ have selector children $v_1$ and $v_2$, and let $v_1$
have leaf children $a$ and $b$. Use PRP with $\eta=1/10$, one context and
uniform starts. Table~\ref{tab:example} runs two rounds. In round~1 the
path is $r\to v_1\to a$ and the response is $1$. Node $v_1$ is told
nothing. Its share rises by $0.050$, it decodes $1$, and it raises
$w_{v_1,a}$. In round~2 the path is $r\to v_1\to b$ and the response is
$0$. Its share falls by $0.055$, and it decodes $0$. Node $v_2$ is
inactive in both rounds, yet its share moves opposite to the response
each time. This is why the protocol is gated on activity.

\begin{table}[t]
\caption{Two rounds of PRP with $\eta=1/10$. Values are exact and shown
after each round.}
\label{tab:example}
\centering
\begin{tabular}{lrrr}
\toprule
& Initial & Round 1 & Round 2 \\
\midrule
Response $o_t$ & -- & $1$ & $0$ \\
$w_{r,v_1}$ & 0.500 & 0.550 & 0.495 \\
$w_{r,v_2}$ & 0.500 & 0.450 & 0.505 \\
$\delta_{v_1}$, decoded $s_{v_1}$ & -- & $+0.050$, 1 & $-0.055$, 0 \\
$\delta_{v_2}$ (inactive) & -- & $-0.050$ & $+0.055$ \\
$w_{v_1,a}$ & 0.500 & 0.550 & 0.595 \\
$w_{v_1,b}$ & 0.500 & 0.450 & 0.405 \\
\bottomrule
\end{tabular}
\end{table}

% ==================================================================
\section{The Proportional Reward-Penalty Rule}
\label{sec:prp}
% ==================================================================

This section analyses PRP at a single selector with fixed response
probabilities. The selector has $N\ge2$ children with
$p_i=P(o_t=1\mid i\text{ selected})$, ordered
$p_1\ge p_2\ge\cdots\ge p_N$. Write $c_i=1-p_i$ for the penalty
probability, $p^*=p_1$, $\Delta=p_1-p_2$ and
$\alpha=\Delta+2(1-p^*)=c_1+c_2$. Assume $p_2<1$, so that $\alpha>0$.
For $N=2$ this excludes the case in which every response is a reward and
the mean drift vanishes at every allocation.

\begin{theorem}[Simplex invariance]\label{thm:integrity}
For any response sequence, PRP keeps $\sum_iw_i(t)=1$ and $w_i(t)\ge0$,
provided $w_{c_t}(t)<1$ at every penalty. If $w(0)$ is interior then
$w(t)$ is interior for all $t$, and the proviso holds automatically. No
projection or renormalisation is needed.
\end{theorem}

\begin{proof}
The reward step is a convex combination of $w$ and a vertex. On penalty,
$\sum_i(w'_i-w_i)=-\eta w_i+\eta w_i\sum_{j\ne i}w_j/(1-w_i)=0$, and
every factor is positive on the interior.
\end{proof}

\subsection{Relation to classical schemes}

\begin{lemma}[Mean drifts]\label{lem:G1}
For interior $w$, define $h_i(w_i)=c_i/(1-w_i)$ and
$\bar h(w)=\sum_jw_jh_j(w_j)$. Then
\begin{align}
\text{PRP:}\quad
\mathbb{E}[\Delta w_i\mid w] &= \eta\,w_i\bigl(\bar h(w)-h_i(w_i)\bigr),
\label{eq:drift-general}\\
L_{R-P}:\quad
\mathbb{E}[\Delta w_i\mid w] &= \eta\Bigl(\tfrac{1}{N-1}\textstyle\sum_{j\ne i}w_jc_j-w_ic_i\Bigr).
\label{eq:drift-lrp}
\end{align}
\end{lemma}

\begin{proof}
For PRP, summing the four cases gives
$\mathbb{E}[\Delta w_i]=\eta w_i[(p_i-w_i)+\sum_{j\ne i}w_j(w_j-p_j)/(1-w_j)]$.
Substituting $p_j-w_j=(1-w_j)(1-h_j)$ and
$w_j(w_j-p_j)/(1-w_j)=-w_j+w_jh_j$ reduces the bracket to $\bar h-h_i$.
For $L_{R-P}$, the four cases give
$w_ip_i(1-w_i)-w_i^2c_i-w_i\sum_{j\ne i}w_jp_j
+\sum_{j\ne i}w_jc_j(\tfrac{1}{N-1}-w_i)$. Collecting the terms in $w_i$
with $\sum_jw_j=1$ leaves~\eqref{eq:drift-lrp}.
\end{proof}

\begin{proposition}[Relation to linear reward-penalty]\label{prop:classical}
\emph{(a)} For $N=2$, PRP and symmetric $L_{R-P}$ are the same rule. The
rest point is $w_1^*=c_2/(c_1+c_2)$, the value given by
McLaren~\cite{McLaren1966}.
\emph{(b)} For $N\ge3$, the interior rest point of the $L_{R-P}$ mean flow
equalises $w_ic_i$, so $w_i\propto1/c_i$. An interior rest point of the
PRP mean flow equalises $c_i/(1-w_i)$, so $1-w_i\propto c_i$.
\end{proposition}

\begin{proof}
(a) With $N=2$ and $j\ne i$, \eqref{eq:neg-sibling} gives
$w'_j=w_j+\eta w_i$ because $w_j=1-w_i$. $L_{R-P}$ gives
$w'_j=(1-\eta)w_j+\eta=w_j+\eta w_i$. The reward steps coincide. The rest
point follows from Theorem~\ref{thm:equilibrium}(a).
(b) \eqref{eq:drift-lrp} vanishes for all $i$ when $w_ic_i$ is constant.
\eqref{eq:drift-general} vanishes on the interior when $h_i$ is constant.
\end{proof}

The two rules weigh failure differently. Take $c=(0.1,0.3,0.5)$.
$L_{R-P}$ settles at $(0.652,0.217,0.130)$ and keeps every action. PRP
settles at $(0.75,0.25,0)$ and drops the worst
(Theorem~\ref{thm:G7}(b)). In this instance PRP puts more weight on the
best action than $L_{R-P}$ does, and still splits weight between the two
actions it retains.

\subsection{Equilibrium}

\begin{theorem}[Equilibrium allocation]\label{thm:equilibrium}
\emph{(a)} For $N=2$,
$\mathbb{E}[\Delta w_1\mid w_1]=\eta\alpha(w_1^*-w_1)$ with
\begin{equation}\label{eq:w-star-exact}
w_1^*=\frac{\Delta+1-p^*}{\Delta+2(1-p^*)},
\end{equation}
which is the unique globally stable fixed point of the expected dynamics.

\emph{(b)} For general $N$, suppose
\begin{equation}\label{eq:interiority}
p_N>\frac{\sum_jp_j-1}{N-1}.
\end{equation}
Then the expected dynamics have the unique interior equilibrium
\begin{equation}\label{eq:explicit-eq}
w_i^*=\frac{p_i+\gamma}{1+\gamma},\qquad
\gamma=\frac{1-\sum_jp_j}{N-1},
\end{equation}
and $p_i>p_j$ implies $w_i^*>w_j^*$.
\end{theorem}

\begin{proof}
(b) On the interior,~\eqref{eq:drift-general} vanishes iff
$h_i=\lambda$ for all $i$. Then $1-w_i=c_i/\lambda$, and summing gives
$\lambda=(N-\sum_ip_i)/(N-1)=1+\gamma$. Hence
$w_i^*=(p_i+\gamma)/(1+\gamma)$ and
$w_i^*-w_j^*=(p_i-p_j)/(1+\gamma)$, which gives uniqueness and order. The
point is interior iff $p_i+\gamma>0$ for all $i$, which
is~\eqref{eq:interiority}.
(a) With $w=w_1$, \eqref{eq:drift-general} gives
$\mathbb{E}[\Delta w]=\eta[c_2-(c_1+c_2)w]=\eta\alpha(w^*-w)$.
\end{proof}

At equilibrium the ratio $(p_i-w_i^*)/(1-w_i^*)=-\gamma$ is the same for
every child. Unrealised quality per unit of remaining share is equalised.
When all qualities are equal, $w_i^*=1/N$.

\begin{theorem}[Stochastic convergence, $N=2$]\label{thm:stochN2}
\emph{(a)} Under steps $\eta_t\in(0,1)$ with $\sum_t\eta_t=\infty$ and
$\sum_t\eta_t^2<\infty$, $w_1(t)\to w_1^*$ almost surely.
\emph{(b)} Under a fixed $\eta$, $\mathbb{E}[w_1(t)]\to w_1^*$ and
$\limsup_t\mathbb{E}[(w_1(t)-w_1^*)^2]\le\eta/(\alpha(2-\eta\alpha))$.
\end{theorem}

\begin{proof}
(a) The recursion is a stochastic approximation with linear drift
$\alpha(w_1^*-w_1)$ and bounded martingale-difference noise. The iterates
stay in $[0,1]$ by Theorem~\ref{thm:integrity}. Theorem~2.1
of~\cite{Borkar2008} applies, and the limiting ODE is a linear
contraction to $w_1^*$.
(b) The centred recursion is
$w_1(t+1)-w_1^*=(1-\eta\alpha)(w_1(t)-w_1^*)+\eta M_{t+1}$ with
$\mathbb{E}[M_{t+1}^2]\le1$. The mean decays geometrically and the second
moment is bounded by $\eta^2/(1-(1-\eta\alpha)^2)$.
\end{proof}

\subsection{Local rate}

\begin{theorem}[Local linear rate]\label{thm:stability}
Under~\eqref{eq:interiority}, the linearisation of the PRP mean flow at
$w^*$ has real, strictly negative eigenvalues on the tangent space of the
simplex. With $d_i=c_i/(1-w_i^*)^2$, the least negative eigenvalue
$\lambda_{\max}$ satisfies
\begin{equation}\label{eq:rate-bound}
\min_iw_i^*d_i\le|\lambda_{\max}|\le\max_iw_i^*d_i,
\end{equation}
where $w_i^*d_i=(p_i+\gamma)(1+\gamma)/(1-p_i)$. For $N=2$ the single
eigenvalue is $-\alpha$.
\end{theorem}

\begin{proof}
Let $F_i(w)=w_i(\bar h-h_i)$, $W=\mathrm{diag}(w^*)$ and
$D=\mathrm{diag}(d)$. For $z$ with $\mathbf{1}^\top z=0$, the variation
of $\bar h$ is $\sum_jz_jh_j+\sum_jw_j^*d_jz_j$, and the first sum
vanishes because all $h_j$ equal $\lambda$ at $w^*$. Hence
$Jz=W[\mathbf{1}(w^{*\top}Dz)-Dz]$. In the inner product
$\langle x,y\rangle=x^\top W^{-1}y$ on the tangent space,
$\langle x,Jy\rangle=-x^\top Dy$. So $J$ is self-adjoint there and
negative definite. The Rayleigh quotient
$-z^\top Dz/z^\top W^{-1}z$ gives~\eqref{eq:rate-bound}.
\end{proof}

\subsection{Global convergence}
\label{sec:global-convergence}

Assume $p_i<1$ for all $i$. Rescale time so that the mean flow is
$\dot w_i=w_i(\bar h-h_i)$.

\begin{lemma}[Potential]\label{lem:G2}
Let $\Phi(w)=-\sum_ic_i\log(1-w_i)$ on $\{w\in\Delta^{N-1}:\max_iw_i<1\}$.
Then $\partial\Phi/\partial w_i=h_i$, $\Phi$ is strictly convex, and
along the mean flow $\dot\Phi=-\mathrm{Var}_w(h)\le0$.
\end{lemma}

\begin{proof}
The Hessian is diagonal with entries $c_i/(1-w_i)^2>0$. The chain rule
and $\sum_iw_i=1$ give
$\dot\Phi=\sum_ih_iw_i(\bar h-h_i)=\bar h^2-\sum_iw_ih_i^2$.
\end{proof}

\begin{lemma}[Information projection]\label{lem:G3}
The slacks sum to $N-1$, so $\sigma_i=(1-w_i)/(N-1)$ is a distribution, as
is $\varphi_i=c_i/C$ with $C=\sum_jc_j$. Then
$\Phi(w)=C\cdot\mathrm{KL}(\varphi\|\sigma)+\text{const}$. The
unconstrained minimiser $\sigma=\varphi$ is $w^*$, and it is interior iff
\eqref{eq:interiority} holds.
\end{lemma}

\begin{proof}
$\Phi=-C\sum_i\varphi_i\log((N-1)\sigma_i)
=C\,\mathrm{KL}(\varphi\|\sigma)+CH(\varphi)-C\log(N-1)$. Interiority
$w_i>0$ reads $\varphi_i<1/(N-1)$, which is
$p_i>(\sum_jp_j-1)/(N-1)$.
\end{proof}

\begin{theorem}[Global convergence]\label{thm:G7}
Let $\widehat w$ minimise $\Phi$ over the closed simplex. From every
interior start the mean flow converges to $\widehat w$.
\emph{(a)} Under~\eqref{eq:interiority}, $\widehat w=w^*$, which is
globally asymptotically stable on the interior.
\emph{(b)} In general
\begin{equation}\label{eq:threshold}
\widehat w_i=\frac{(p_i-\tau)_+}{1-\tau},\qquad
\sum_i(p_i-\tau)_+=1-\tau,
\end{equation}
for a unique threshold $\tau$. Children with $p_i\le\tau$ are dropped.
\end{theorem}

\begin{proof}
\emph{Threshold.} The KKT conditions give a multiplier $\mu$ with
$h_i(\widehat w_i)=\mu$ where $\widehat w_i>0$ and $c_k\ge\mu$ where
$\widehat w_k=0$. With $\tau=1-\mu$ this
is~\eqref{eq:threshold}. Strict convexity gives uniqueness.

\emph{Lyapunov function.} Let
$V(w)=\sum_{i:\widehat w_i>0}\widehat w_i\log(\widehat w_i/w_i)$. Along
the flow,
$\dot V=-\sum_i\widehat w_i(\bar h-h_i)=\nabla\Phi(w)^\top(\widehat w-w)
\le\Phi(\widehat w)-\Phi(w)<0$ for $w\ne\widehat w$, by convexity.

\emph{Convergence.} A finite sublevel set of $V$ is compact and forward
invariant. It bounds every coordinate in the support of $\widehat w$ away
from zero, and that support has at least two elements when all $p_i<1$.
Standard Lyapunov arguments give $w(t)\to\widehat w$.
\end{proof}

Theorem~\ref{thm:G7} is a statement about the mean flow. It gives no
global exponential rate. Near $w^*$ the rate is the spectral gap of
Theorem~\ref{thm:stability}. The stochastic recursion follows it on the
event that the weights stay away from the boundary.

\begin{corollary}[Stochastic convergence, general $N$]\label{cor:stochN}
Run PRP with decreasing steps as in Theorem~\ref{thm:stochN2}(a),
under~\eqref{eq:interiority}. On the event
$\{\liminf_t\min_iw_i(t)>0\}$, $w(t)\to w^*$ almost surely.
\end{corollary}

\begin{proof}
On the event the iterates eventually lie in a compact subset $K$ of the
open simplex. On $K$ the mean field of~\eqref{eq:drift-general} is
Lipschitz and the per-round noise is a bounded martingale difference, so
the stochastic-approximation theorem of Borkar~\cite{Borkar2008} applies
as in the proof of Theorem~\ref{thm:stochN2}(a). The limit set is a
compact, connected, internally chain-transitive invariant set of the mean
flow in $K$. The function $V$ in the proof of Theorem~\ref{thm:G7}
decreases strictly along the flow away from $w^*$, and $w^*$ is the only
equilibrium in the open simplex (Theorem~\ref{thm:equilibrium}), so the
only such set is $\{w^*\}$.
\end{proof}

\begin{remark}[The interiority event]\label{rem:interior-event}
The event in Corollary~\ref{cor:stochN} is expected to have probability
one under~\eqref{eq:interiority}, since every face and vertex equilibrium
of the mean flow is then repelling. At a face equilibrium with excluded
set $E$, the drift of $\log w_k$ for $k\in E$ equals the face multiplier
minus $c_k$, and~\eqref{eq:interiority} makes that difference positive
for every $k\in E$. Near a vertex the drift of every vanishing weight is
positive and unbounded. Ruling out almost-sure convergence to a repelling
equilibrium needs control of the fluctuations of $\log w_k$, in the
manner of non-convergence results for stochastic approximation, and is
left open. A weight floor, as in the tracking study of
Section~\ref{sec:tracking}, enforces the event directly, at the price of
a projection step whose limiting dynamics are not analysed here.
\end{remark}

\begin{figure}[t]
\centering
\includegraphics[width=\columnwidth]{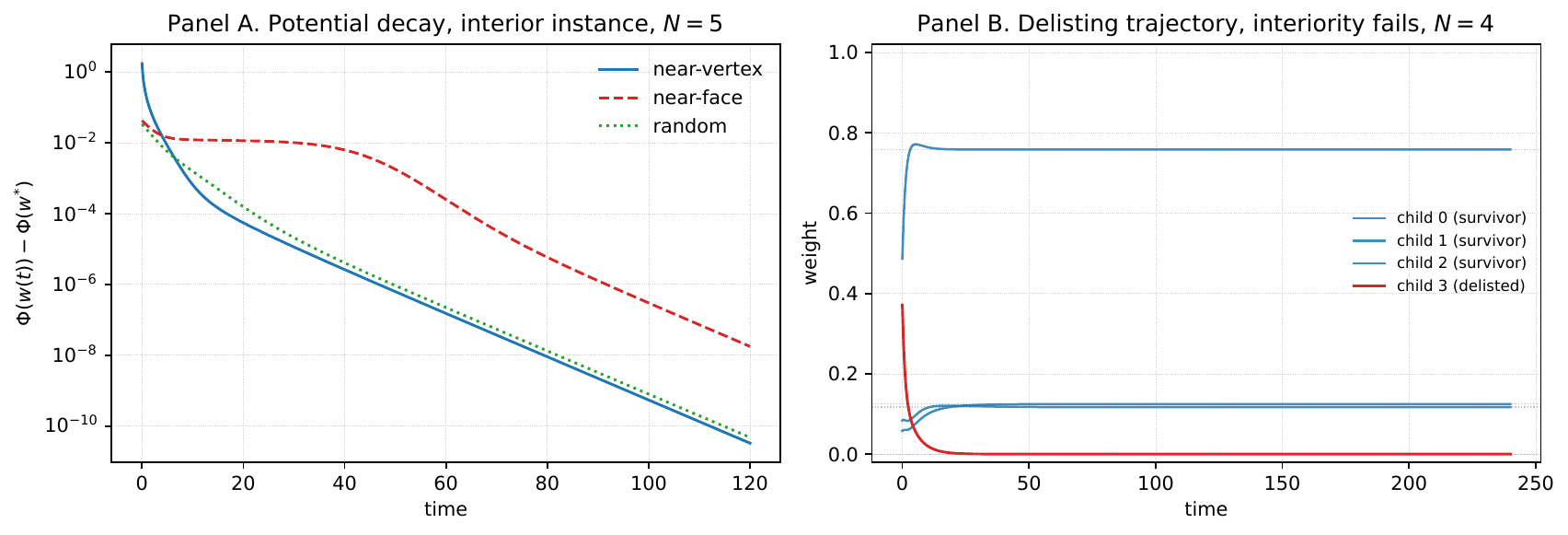}
\caption{Global convergence of the PRP mean flow. (A) Potential gap from
three starts on a five-child instance that satisfies the interiority
condition. (B) A four-child instance that violates it. The
below-threshold child's weight goes to zero and the survivors reach the
threshold allocation~\eqref{eq:threshold}.}
\label{fig:global-convergence}
\end{figure}

\begin{remark}[Equilibrium cost]\label{rem:eq-cost}
The per-round loss from not always selecting the best child is
$c_{\mathrm{eq}}=\sum_{j\ge2}(p_1-p_j)w_j^*$. For $N=2$,
$c_{\mathrm{eq}}=\Delta(1-p^*)/(\Delta+2(1-p^*))$, which is increasing in
$\Delta$. So $c_{\mathrm{eq}}/p^*\le(1-p^*)/(2-p^*)$, and for
$p^*\ge0.95$ the fractional loss is at most $1/21\approx4.76\%$. The loss
does not depend on $\eta$. A decreasing step converges to the same
$w^*$.
\end{remark}

\begin{remark}[Dropping and recovery]\label{rem:lockout}
The drift~\eqref{eq:drift-general} is proportional to $w_i$. A child
driven towards zero by Theorem~\ref{thm:G7}(b) therefore recovers slowly
if its quality later rises. Section~\ref{sec:tracking} shows that without
a floor the recovery does not happen on practical horizons. A weight
floor removes the problem.
\end{remark}

% ==================================================================
\section{Composition Along the Active Path}
\label{sec:hierarchical}
% ==================================================================

\begin{theorem}[Marginal composition]\label{thm:hierarchy}
Let selector $v$ have fixed child response probabilities on its
active-round subsequence.
\emph{(a)} On an active round, $v$ updates from its own weights, its
selected child and $s_v=o_t$. Its transition map is the standalone
single-selector map of Section~\ref{sec:prp}.
\emph{(b)} The probability that $v$ is active is
$P(A_{v,t}\mid\mathcal{F}_{t-1},x_t)=\prod_{(u,c)\in\mathrm{path}(r,v)}
w_{u\to c}^{\kappa_u(x_t)}(t)$. Ancestors may change the distribution of
inputs on $v$'s active subsequence, but not the identity $s_v=o_t$.
\emph{(c)} A leaf's selection probability is the product of the
weights on its path.
\end{theorem}

\begin{proof}
(a) follows from Theorem~\ref{thm:signal} and C3. (b) and (c) follow from
the chain rule of conditional probability, since a node is active iff
every ancestor selects the next node on its path. Part~(c) is equation
(A1) of~\cite{ThathacharRamakrishnan1981}. The factors may be
statistically dependent, because the selectors share inputs and
responses and their clocks are nested.
\end{proof}

Theorem~\ref{thm:hierarchy} is not a joint convergence result. It says
that every active round carries the correct signal from round one, so
deeper levels start learning before their ancestors settle. Ancestors
leave $v$'s child response probabilities on a cell unchanged when those
probabilities are constant across the ancestors' cells inside it, as
when $v$'s partition refines every ancestor's partition on its path, or
the cell already carries everything $v$'s response depends on. Otherwise the fixed-response hypothesis holds only
once the ancestors have settled.

% ==================================================================
\section{Faulty Intermediaries}
\label{sec:distortion}
% ==================================================================

An intermediate selector receives one bit and acts on one bit. A fault
or a misconfiguration can replace the identity map by one of the other
three deterministic maps $\{0,1\}\to\{0,1\}$, namely stuck-at-0,
stuck-at-1 or inversion. This section asks how much quality a selector
then delivers upward. For the equilibrium comparisons, assume $0<q_i<1$.

\begin{lemma}[Pairwise covariance identity]\label{lem:pwcov}
For any probability vector $w$ and vector $q$,
\begin{equation}\label{eq:pwcov}
\sum_iw_iq_i-\bar q=\frac1N\sum_{i<j}(w_i-w_j)(q_i-q_j),\qquad
\bar q=\frac1N\sum_iq_i.
\end{equation}
\end{lemma}

\begin{proof}
Expand the right side. Each $w_kq_k$ appears in $N-1$ pairs and each
cross term $w_iq_j$ once. Collecting terms with $\sum_iw_i=1$ gives the
left side.
\end{proof}

\begin{lemma}[Monotone response]\label{lem:global-monotone}
Hold the other children's qualities fixed and vary $q_v\in(0,1)$. The
equilibrium weight $\widehat w_v(q_v)$ of Theorem~\ref{thm:G7} is
continuous and nondecreasing, and strictly increasing wherever it is
positive. On a support $A\ni v$ with $m=|A|\ge2$,
\begin{equation}\label{eq:face-derivative}
\frac{\partial w_v^A}{\partial q_v}
=\frac{(m-1)-\sum_{j\in A,j\ne v}q_j}{(m-1)(1+\gamma_A)^2}>0,
\end{equation}
with $\gamma_A=(1-\sum_{j\in A}q_j)/(m-1)$.
\end{lemma}

\begin{proof}
On $A$, $w_v^A=(q_v+\gamma_A)/(1+\gamma_A)$ and
$\partial\gamma_A/\partial q_v=-1/(m-1)$. Differentiating and using
$(m-1)(1+\gamma_A)=m-\sum_{j\in A}q_j$ gives~\eqref{eq:face-derivative}.
The numerator is positive because each $q_j<1$. At a change of support
the entering or leaving child's weight is zero, so the response is
continuous.
\end{proof}

\begin{definition}[Delivered quality]\label{def:delivered-quality}
Fix selector $v$ with fixed child qualities $q_1,\ldots,q_N$. Let $v$
apply $f\in\{\mathrm{id},\mathrm{neg},\mathrm{const}\}$ to its bit before
it updates. Write $Q_v(f)=\sum_i\widehat w_{v,i}(f)\,q_i$, where
$\widehat{\mathbf w}_v(f)$ is the equilibrium allocation that $f$
induces. Identity presents $q$ and inversion presents $1-q$. Under either
constant every child presents the same apparent quality, so from a
uniform start the weight process is exchangeable across children and its
expected weights are uniform at every round. For the constant maps
$\widehat{\mathbf w}_v$ denotes this expected allocation.
\end{definition}

\begin{corollary}[Fault ordering]\label{lem:qual-ordering}
If the $q_i$ are not all equal, then
$Q_v(\mathrm{id})>\bar q=Q_v(\mathrm{const})>Q_v(\mathrm{neg})$.
\end{corollary}

\begin{proof}
Under identity the allocation~\eqref{eq:threshold} is comonotone with
$q$, strictly for some pair, so~\eqref{eq:pwcov} is positive. Under a
constant the expected allocation is uniform, so~\eqref{eq:pwcov}
vanishes. Under inversion the allocation is comonotone with $1-q$, so
\eqref{eq:pwcov} is negative.
\end{proof}

\begin{corollary}[Dispersion]\label{cor:dispersion}
Under~\eqref{eq:interiority},
\begin{equation}\label{eq:dispersion}
Q_v(\mathrm{id})-\bar q=\frac{\sum_{i<j}(q_i-q_j)^2}{N(1+\gamma)}\ge0.
\end{equation}
\end{corollary}

\begin{proof}
Substitute $w_i-w_j=(q_i-q_j)/(1+\gamma)$ into~\eqref{eq:pwcov}.
\end{proof}

\begin{remark}[Scope]\label{rem:distortion-scope}
The ordering assumes that $v$'s parent decodes honestly, and it compares
equilibrium responses at one selector. Several faulty levels at once
interact, since a fault changes the bit that the children receive. That
case is open. Randomised faults are not separated from the identity. A
channel that turns each penalty into a reward with probability $b\in[0,1)$ maps
$q_i$ to $b+(1-b)q_i$ and $\tau$ to $b+(1-b)\tau$, which leaves
\eqref{eq:threshold} unchanged.
\end{remark}

% ==================================================================
\section{Numerical Results}
\label{sec:experiments}
% ==================================================================

\subsection{High-precision paired replay}
\label{sec:replay}

Synthetic hierarchies with branching $b\in\{2,3,4\}$ cover 22
branching-depth cells and up to 16,384 leaves. Each run compares
\emph{protocol mode} (non-root nodes decode from their share change) with
\emph{broadcast mode} (every active node receives $o_t$). The paired
replay uses PRP with $\eta=0.1$, two contexts per selector that alternate
by round, one random stream per pair, and decimal arithmetic with 300
significant digits. Selection is $\epsilon$-greedy with an exploration
probability that decays linearly from 0.15 to zero over the first 3,000
of 5,000 rounds. The response thresholds the selected leaf's quality plus
Gaussian noise of standard deviation 0.08 at 0.5. The replay has no
decoding threshold, skipped update or projection. Every selector uses
the same two-context map, so differing maps are not exercised.

Across 1,100,000 paired traversals all 4,150,000 decoded bits match,
all 5,250,000 local pre- and post-update vectors match, and all 220 final
state hashes match (Table~\ref{tab:depth-scaling}). These results agree with
Proposition~\ref{prop:equivalence}. An earlier float64
implementation recorded 33,932 mismatches (0.818\%), all where a
boundary-adjacent change rounded to zero or the represented state left
the simplex. Figure~\ref{fig:depth-scaling} shows the ratio of protocol
to broadcast accuracy.

\begin{table}[t]
\caption{High-precision paired replay (5,000 rounds per seed, 10 seeds,
$\eta=0.1$, 300-digit decimals). Ratio is protocol over broadcast leaf
accuracy.}
\label{tab:depth-scaling}
\centering
\small
\begin{tabular}{@{}crrrr@{}}
\toprule
$b$ & Depths & Decoded bits & Mismatch & Ratio \\
\midrule
2 & 2--8, 10, 12, 15 & 2,600,000 & 0 & 1.000 \\
3 & 2--8 & 1,050,000 & 0 & 1.000 \\
4 & 2--6 & 500,000 & 0 & 1.000 \\
\midrule
Total & 22 cells & 4,150,000 & 0 & 1.000 \\
\bottomrule
\end{tabular}
\end{table}

\begin{figure}[t]
\centering
\includegraphics[width=\columnwidth]{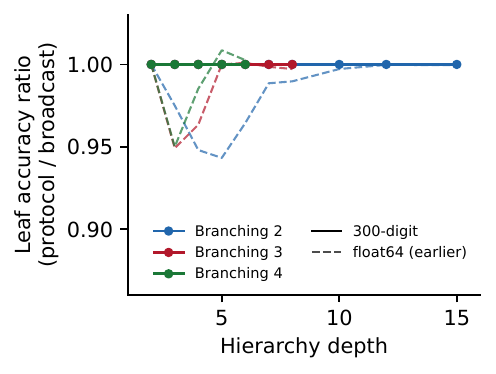}
\caption{Protocol-to-broadcast leaf accuracy ratio by depth and branching.
The ratio is 1.0 in every cell with 300-digit arithmetic. The dashed line is
the earlier float64 implementation.}
\label{fig:depth-scaling}
\end{figure}

\subsection{Natural hierarchies}

Three public datasets supply institutional hierarchies. They are U.S.\
Census PUMS~\cite{ACSPUMS2022} (475 leaves, depth 4), PISA
2022~\cite{PISA2022} (1,567 leaves, depth 4) and the
S\&P~500~\cite{SP500Data2026} (397 leaves, depth 3). Each runs PRP with
$\eta=0.05$ for 50,000,000 rounds per seed over 10 seeds in float64, with
no exploration. Each leaf's response probability is its dataset statistic
mapped by rank onto $[0.05,0.95]$, and the response is a Bernoulli draw
with that probability. On penalty the float64 update scales the siblings
by their represented mass rather than by $1-w_i$, and every update ends
by renormalising. In exact arithmetic this is the same map
(Theorem~\ref{thm:integrity}). In floating point it keeps rounding error
in the sum from compounding. The
decoded bit matched the root response in all 3,402,449,134 active node
comparisons (Table~\ref{tab:natural-hierarchy}). This implementation
treats a change below $10^{-12}$ in magnitude as no signal and skips the
child's update as well as the comparison. No such case occurred.

\begin{table}[t]
\caption{Field-scale decoding check (50M rounds per seed, 10 seeds,
float64).}
\label{tab:natural-hierarchy}
\centering
\small
\begin{tabular}{@{}lrrrr@{}}
\toprule
Dataset & Leaves & Depth & Active comps. & Mism. \\
\midrule
Census PUMS & 475 & 4 & 1,500,000,000 & 0 \\
PISA 2022 & 1,567 & 4 & 902,449,134 & 0 \\
S\&P 500 & 397 & 3 & 1,000,000,000 & 0 \\
\midrule
Total & 2,439 & & 3,402,449,134 & 0 \\
\bottomrule
\end{tabular}
\end{table}

\subsection{Tracking a changing environment}
\label{sec:tracking}

The third study asks which rule to run under the protocol when qualities
change. The tree has depth 3 and branching 3, so 27 leaves, with
$\eta=0.01$ and 20 seeds. The best leaf has quality 0.97. The others are
drawn once from $U(0.3,0.65)$. The response is a Bernoulli draw with the
selected leaf's quality, selection samples the weights with no
exploration, and a child decodes $\mathbf{1}[\delta_v>0]$ as in
Definition~\ref{def:protocol}. The run lasts 200,000 rounds in equal
phases. At each switch the best leaf moves to a different top-level
subtree, and the old best falls to 0.3, so every level must re-route.
PRP, $L_{R-P}$ and $L_{R-I}$ each run under the protocol, both pure and
with a weight floor of 0.01 applied by clamping and renormalising. PRP
uses the float64 update of the natural-hierarchy runs.

\begin{figure}[t]
\centering
\includegraphics[width=\columnwidth]{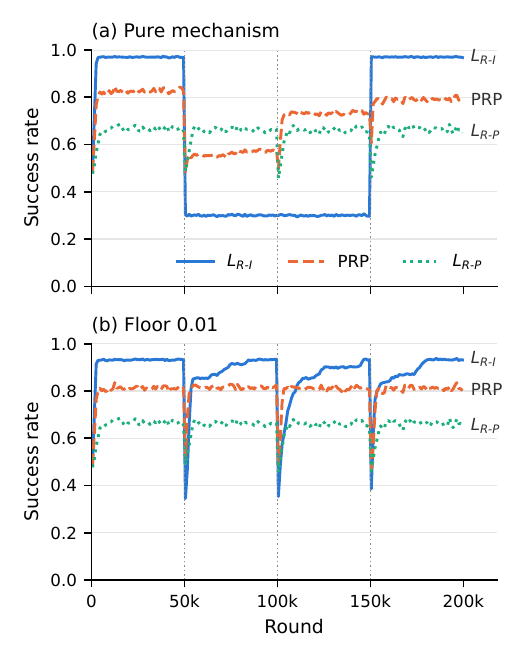}
\caption{Success rate over rounds in the tracking study, mean over 20
seeds in 1,000-round bins, protocol mode. Dotted verticals mark the phase
switches. (a) Pure mechanism. (b) Weight floor 0.01.}
\label{fig:tracking}
\end{figure}

\begin{table}[t]
\caption{Mean success rate over 200,000 rounds, four phases, 20 seeds.
Protocol mode throughout. Broadcast mode in brackets where the rounded
value differs.}
\label{tab:tracking}
\centering
\small
\begin{tabular}{@{}lrr@{}}
\toprule
Rule & Pure & Floor 0.01 \\
\midrule
PRP & 0.722 & 0.801 \\
$L_{R-P}$ & 0.654 & 0.654 \\
$L_{R-I}$ & 0.631 & 0.877 (0.880) \\
\bottomrule
\end{tabular}
\end{table}

\begin{table}[t]
\caption{Mean success rate with floor 0.01 as the switch frequency rises
(200,000 rounds, 20 seeds).}
\label{tab:switch}
\centering
\small
\begin{tabular}{@{}rrrr@{}}
\toprule
Phase length & PRP & $L_{R-P}$ & $L_{R-I}$ \\
\midrule
20,000 & 0.780 & 0.634 & 0.852 \\
10,000 & 0.744 & 0.603 & 0.816 \\
5,000 & 0.676 & 0.549 & 0.751 \\
2,500 & 0.571 & 0.495 & 0.633 \\
\bottomrule
\end{tabular}
\end{table}

Fig.~\ref{fig:tracking} shows the success rate over rounds, and
Tables~\ref{tab:tracking} and~\ref{tab:switch} give the means. Five
findings follow.

\emph{The protocol preserves the learning trajectory.} In exact arithmetic,
Proposition~\ref{prop:equivalence} gives identical trajectories for
sign-separating rules. In these float64 runs, PRP and $L_{R-P}$ give the
same overall success rates under protocol and broadcast, to three
decimals, in both regimes. The unrounded means with the floor are
0.80133450 versus 0.80141925 for PRP and 0.87685850 versus 0.87969625
for $L_{R-I}$. These differences are 0.008 and 0.284 percentage points,
respectively; $L_{R-P}$ is unchanged. Clamping and renormalising can erase
a reward increment or raise a selected share on penalty at the floor or
cap, which breaks sign separation there. Table~\ref{tab:mismatch}
classifies the resulting decoding errors. These measured effects concern the chosen
floor implementation. An implementation preserving sign separation
retains the exact decoding guarantee.

\emph{Floored reward-inaction tracks best.} With the floor, $L_{R-I}$
beats PRP at every switch frequency tested. Pure $L_{R-I}$ locks in. After
a switch it keeps the old path, and success equals the fallen leaf's 0.3
until the schedule returns to that leaf.

\emph{Pure PRP locks out.} Subtrees that PRP has dropped do not come back
(Remark~\ref{rem:lockout}). In a 400,000-round two-phase run, success in
the second phase stayed between 0.499 and 0.504 in every 20,000-round
window. The floor removes the effect.

\emph{PRP beats $L_{R-P}$ everywhere.} PRP scores higher in both regimes
and at every switch frequency, which is consistent with the sharper
allocation of Proposition~\ref{prop:classical}(b).

\emph{Where decoding fails.} The decoding errors are classified in
Table~\ref{tab:mismatch}. Pure $L_{R-I}$ drives shares to within
$10^{-12}$ of one, where the reward increment is lost to rounding. Each of its
4,880,604 errors occurs there, and the protocol and broadcast runs still
reach the same success rate, 0.631. With the floor, errors occur
when the share sits at the cap or the floor, and in 181 penalty rounds in
which renormalisation raised the share by at most $1.1\times10^{-16}$. Pure PRP
and $L_{R-P}$ in both regimes record no errors in 8,000,000 decodes
each. Floored PRP records 132 errors in the broadcast-trajectory audit.

\begin{table}[t]
\caption{Classification of decoding errors in the tracking study,
8,000,000 decodes per configuration, counted along the broadcast
trajectory.}
\label{tab:mismatch}
\centering
\small
\begin{tabular}{@{}lrrr@{}}
\toprule
Configuration & At cap & At floor & One-ulp rise \\
\midrule
$L_{R-I}$, pure & 4,880,604 & 0 & 0 \\
$L_{R-I}$, floor & 1,432,749 & 45,884 & 181 \\
PRP, floor & 105 & 27 & 0 \\
PRP pure; $L_{R-P}$ both & 0 & 0 & 0 \\
\bottomrule
\end{tabular}
\end{table}

% ==================================================================
\section{Discussion}
\label{sec:discussion}
% ==================================================================

\subsection{Choosing a rule}
The protocol is agnostic to the rule, and the choice of rule is a
separate question from the results of Sections~\ref{sec:protocol}
and~\ref{sec:prp}. Those results concern whether the response is
recovered exactly and which allocation PRP converges to. They say nothing
about how closely a hierarchy follows a moving best leaf, which is the
question the tracking study of Section~\ref{sec:tracking} asks. On that
question a floored $L_{R-I}$ hierarchy follows the best leaf more closely
than PRP or $L_{R-P}$ in the environments tested.
For the unmodified $L_{R-I}$ hierarchy, Corollary~\ref{cor:tp95} transfers
the established ODE and local convergence results under their original
assumptions. The floored tracking experiment tests a different
implementation in a changing environment. PRP fits a different requirement.
It keeps an interior allocation over the actions above a computable
threshold, with a closed
form for the operating point and a global potential. Where a deployment
must keep traffic spread over several providers, PRP sets that spread
directly from their failure rates. Learning automata have also been used
to balance load across nodes so that their performance levels
equalise~\cite{Yazidi2021}.

\subsection{Open problems}
Joint convergence of a PRP hierarchy whose levels adapt at once is open.
The method of~\cite{ThathacharPhansalkar1995} is the natural template. It
uses expected reward as the Lyapunov function, and that will not work
here, since the PRP rest point is not a reward maximum. A candidate is a
sum of the per-node potentials of Lemma~\ref{lem:G3}, weighted by
activation probabilities. Further open questions are a global rate for
Theorem~\ref{thm:G7}, removal of the interiority event from
Corollary~\ref{cor:stochN}, and the protocol when the sign is observed
with noise, where C4 is relaxed. Convergence when selectors change their
maps during the run, and the ODE of Corollary~\ref{cor:tp95} when the
maps differ across levels, are also open. A further question is whether a child
can recover the response under estimator and pursuit
updates~\cite{Yazidi2020,Omslandseter2024}, which are not
sign-separating.

\subsection{Scope}
The model selects one child per selector and one path per round. It does
not divide a shared outcome among simultaneous contributors or infer the
quality of unselected children. Those problems need a richer observation
model than C2. Automata that learn an arbitrary ranked subset of
actions~\cite{Zhang2024} solve a different problem and are not covered
here.

% ==================================================================
\section{Conclusion}
\label{sec:conclusion}
% ==================================================================

A hierarchy of learning automata does not need the environment response
to be broadcast. Under any sign-separating rule an active child reads the
response from the sign of its own share change, exactly and at every
depth. The protocol hierarchy and the broadcast hierarchy have the same
sample paths, so their convergence theory is shared. This makes
hierarchical selection usable wherever no feedback channel exists between
adjacent levels, since one bit, the sign of the child's own allocation
change, is the only evaluation information that crosses. The proportional reward-penalty
rule adds an interior allocation with a closed form and a global
potential, and the protocol serves any sign-separating rule.

% ------------------------------------------------------------------
\section*{Data Availability}
Source code, numerical artefacts and analysis scripts for the replay,
natural-hierarchy and tracking-study results are archived under the all-versions DOI
\url{https://doi.org/10.5281/zenodo.22111935}~\cite{OneBitCode2026}.
The tracking-study scripts and results are in
\texttt{experiments/nonstationary\_tracking/} in the archive. The three
source datasets are public and are not redistributed. Retrieval
instructions are in \texttt{supplementary/DATA.md} in the archive.

\bibliographystyle{unsrt}
\bibliography{onebit-preprint}
\end{document}